\documentclass[a4paper]{jpconf}
\usepackage{amsthm, amssymb, amsmath} 

\newcommand{\Cl}{\mathop{\mathrm{Cl}}}

\newtheorem{thm}{Theorem}[section]
\newtheorem{cor}[thm]{Corollary}
\newtheorem{lem}[thm]{Lemma}
\newtheorem{prop}[thm]{Proposition}
\newtheorem{defn}[thm]{Definition}
\newtheorem{ex}{Example}

\begin{document}
\title{Rank-3 root systems induce  root systems of rank 4 via a new Clifford spinor construction}

\author{Pierre-Philippe Dechant}

\address{Department of Mathematics, University of Durham,\\ Institute for Particle Physics Phenomenology, Ogden Centre for Fundamental Physics, Department of Physics, University of Durham, South Rd, Durham DH1 3LE, UK and \\Department of Mathematics, York Centre for Complex Systems Analysis, University of York, Heslington, York YO10 5GG, UK}

\ead{ppd22@cantab.net}

\begin{abstract}
	In this paper, we show that via a novel construction every rank-3 root system induces a root system of rank 4. 
	Via the Cartan-Dieudonn\'e theorem, an even number of successive Coxeter reflections yields rotations that in a Clifford algebra framework  are described by spinors.
	In three dimensions these spinors themselves have a natural four-dimensional Euclidean structure, and discrete spinor groups can therefore be interpreted as  4D polytopes.
	In fact, we show that these polytopes have to be root systems, thereby inducing Coxeter groups of rank 4, and that their automorphism groups include two factors of the respective discrete spinor groups trivially acting on the left and on the right by spinor multiplication.
	Special cases of this general theorem include the exceptional 4D groups $D_4$, $F_4$ and $H_4$, which therefore opens up a new understanding of applications of these structures in terms of spinorial geometry. In particular, 4D groups are ubiquitous in high energy physics.
	For the corresponding case in two dimensions, the groups $I_2(n)$ are shown to be self-dual, whilst
	 via a similar construction in terms of octonions  each rank-3 root system induces a root system in dimension 8;  this root system is in fact the direct sum of two copies of the corresponding induced 4D root system.
\end{abstract}

\section{Introduction}\label{HGA_intro}


Root systems are useful mathematical abstractions, which are polytopes that generate reflection (Coxeter) groups. Certain families of root systems exist in any dimension, whereas others -- exceptional ones -- only exist as accidental structures in specific dimensions. Root systems in different dimensions are largely thought to be independent of each other (with the exception of sub-root systems). The Freudenthal-Tits magic square \cite{Baez2001Octonions} makes some non-trivial connections, but geometric insight as to why these should exist is scarce. In this paper, we present a novel connection between root systems in different dimensions that has a geometric origin. 

Clifford's Geometric Algebra provides a mathematical framework that generalises the more familiar vector space and matrix methods. In this setup, orthogonal transformations are encoded (in fact, doubly covered) by {versors}, which are the Clifford {geometric product} of several unit vectors, via a {sandwiching prescription}. In particular, the rotations (i.e. the special orthogonal group) are doubly covered by geometric products of an even number of unit vectors, resulting in {spinors}, or {`rotors'}.  These elements in the even-subalgebra  can themselves have a Euclidean structure and can thus be reinterpreted as vectors in a different space. Here, we show that such a construction can induce a root system (in the same or another dimension) from a given root system. This systematises the observations made in \cite{Dechant2012CoxGA, Dechant2012AGACSE, Dechant2013Platonic},
and opens up a new -- spinorial -- view of the geometry of root systems (where by spinorial we mean something related to the double cover of the rotations i.e. special orthogonal transformations), with interesting applications particularly concerning the interplay of three and four dimensions, and notably the exceptional root systems $D_4$, $F_4$ and $H_4$:  $D_4$, the root system related to $SO(8)$ with exceptional triality symmetry, $F_4$, the largest crystallographic group in 4D, and  the largest non-crystallographic group, $H_4$.

The rest of this paper is organised as follows. Section \ref{HGA_Cox} introduces Coxeter groups and root systems, and Section \ref{HGA_GA} gives the necessary  Clifford algebra background. The Induction Theorem is stated and proven in Section \ref{HGA_4D}, and the self-duality of the two-dimensional case is discussed in Section \ref{HGA_2D}. We extend our approach to a related octonionic construction in Section \ref{HGA_Oct}.
Conclusions are given in Section \ref{HGA_Concl}.

\section{Coxeter Groups}\label{HGA_Cox}

\begin{defn}[Coxeter group] A {Coxeter group} is a group generated by some involutive generators $s_i, s_j \in S$ subject to relations of the form $(s_is_j)^{m_{ij}}=1$ with $m_{ij}=m_{ji}\ge 2$ for $i\ne j$. 
\end{defn}
The  finite Coxeter groups have a geometric representation where the involutions are realised as reflections at hyperplanes through the origin in a Euclidean vector space $\mathcal{E}$ (essentially  just the classical reflection groups). In particular, let $(\cdot \vert \cdot)$ denote the inner product in $\mathcal{E}$, and $\lambda$, $\alpha\in\mathcal{E}$.  
\begin{defn}[Reflections and roots] The generator $s_\alpha$ corresponds to the {reflection}
\begin{equation}\label{reflect}
s_\alpha: \lambda\rightarrow s_\alpha(\lambda)=\lambda - 2\frac{(\lambda|\alpha)}{(\alpha|\alpha)}\alpha
\end{equation}
 at a hyperplane perpendicular to the  {root vector} $\alpha$.
\end{defn}

The action of the Coxeter group is  to permute these root vectors, and its  structure is thus encoded in the collection  $\Phi\in \mathcal{E}$ of all such roots, which form a root system: 
\begin{defn}[Root system] \label{DefRootSys}
{Root systems} are defined by the  two axioms
\begin{enumerate}
\item $\Phi$ only contains a root $\alpha$ and its negative, but no other scalar multiples: $\Phi \cap \mathbb{R}\alpha=\{-\alpha, \alpha\}\,\,\,\,\forall\,\, \alpha \in \Phi$. 
\item $\Phi$ is invariant under all reflections corresponding to vectors in $\Phi$: $s_\alpha\Phi=\Phi \,\,\,\forall\,\, \alpha\in\Phi$.
\end{enumerate}
\end{defn}

Root systems and their associated Coxeter groups describe polyhedral symmetries such as the symmetries of the Platonic Solids, and  are also central in Lie Theory \cite{Kac1994InfDimLA, FuchsSchweigert1997}.
 They have manifold practical applications, for instance to crystals and quasicrystals \cite{Twarock:2002a, Koca:1998, DechantTwarockBoehm2011E8A4}, 
the structure of viruses  \cite{DechantTwarockBoehm2011H3aff} and other polyhedral objects 
\cite{Twarock:2002b, Wardman2011CarbonOnion, Koca2007Polyhedra, Koca2011Chiral}, 
as well as high energy theory 
\cite{Gross1985HE, DamourHenneauxNicolai2002E10, Koca2001H4E8} 
and singularities \cite{Shcherbak:1988,HennauxPersson2008SpacelikeSingularitiesAndHiddenSymmetriesofGravity}.
A recent characterisation and visualisation of the relevant groups in two and three dimensions in a conformal framework was given in \cite{Hestenes2002PointGroups, Hestenes2002CrystGroups, Hitzer2010CLUCalc}.
Some of the 4-dimensional root systems have unusual accidental properties, such as triality in $D_4$ (or better known in the physics literature as the triality between vector and spinor representations of $SO(8)$, important in showing the equivalence of the Green-Schwarz and Ramond-Neveu-Schwarz strings), the largest non-crystallographic Coxeter group $H_4$ (which contains the GUT group $A_4=SU(5)$ as a subgroup), or the largest symmetry group $F_4$ of a  lattice in four dimensions (which also contains the little groups of String and M-Theory, $D_4$ and $B_4$, respectively). These peculiarities are discussed in more detail elsewhere 
\cite{Dechant2012CoxGA, Dechant2013Platonic}
 -- here, we stress the general nature of our argument without the need to discuss individual root systems and their applications.

\section{Geometric Algebra}\label{HGA_GA}

The study of Clifford algebras and Geometric Algebra originated with Grassmann's, Hamilton's and Clifford's geometric work 
\cite{Grassmann1844LinealeAusdehnungslehre, Hamilton1844, Clifford1878}.
 However, the geometric content of the algebras was soon lost when interesting algebraic properties were discovered in mathematics, and Gibbs advocated the use of the hybrid system of vector calculus 
 in physics. When Clifford algebras resurfaced in physics in the context of quantum mechanics, it was purely for their algebraic properties, and this continues in particle physics to this day. Thus, it is widely thought that Clifford algebras are somehow intrinsically quantum mechanical in nature. The original geometric meaning of Clifford algebras has been revived in the work of David Hestenes \cite{Hestenes1966STA, HestenesSobczyk1984, Hestenes1990NewFound}.
 Here, we follow an exposition along the lines of  \cite{LasenbyDoran2003GeometricAlgebra}.

In a manner reminiscent of complex numbers carrying both real and imaginary parts in the same algebraic entity, one can consider the 
geometric product of two vectors defined as the sum of their scalar (inner/symmetric) product and  wedge (outer/ exterior/antisymmetric) product
\begin{equation}\label{in2GP}
    ab\equiv a\cdot b + a\wedge b.
\end{equation}
The wedge product is the outer product introduced by Grassmann, as an antisymmetric product of two vectors, which  naturally defines a plane. Unlike the constituent inner and outer products, the geometric product is invertible, as $a^{-1}$ is simply given by $a^{-1}=a/(a^2)$. This leads to many algebraic simplifications over standard vector space techniques, and also feeds through to the differential structure of the theory, with Green's function methods that are not achievable with vector calculus methods.
This geometric product can be extended to the product of more vectors via associativity and distributivity, resulting in higher grade objects called multivectors.   There are a total of $2^n$ elements in the algebra, since it truncates at grade $n$ multivectors due to the scalar nature of the product of  parallel vectors and the antisymmetry of orthogonal vectors. Essentially, a Clifford algebra is a deformation of the exterior algebra by a quadratic form, and for a Geometric Algebra this is the metric of space(time).

The geometric product provides a very compact and efficient way of handling reflections in any number of dimensions, and thus by the Cartan-Dieudonn\'e theorem also rotations  \cite{Garling2011}.  For a unit vector $n$, we consider the reflection of a vector $a$ in the hyperplane orthogonal to $n$. Thanks to the geometric product in Clifford algebra the two terms in Eq. (\ref{reflect}) combine into a single term and thus a `sandwiching prescription': 
\begin{thm}[Reflections]\label{HGA_refl}
In Geometric Algebra, a vector `$a$' transforms under a reflection in the (hyper-)plane defined by a unit normal vector `$n$' as
	\begin{equation}\label{in2refl}
	  a'=-nan.
	\end{equation}
\end{thm}

This is a remarkably compact and simple prescription for reflecting vectors in hyperplanes. More generally, higher grade multivectors of the form $M= ab\dots c$ (so-called versors) transform similarly (`covariantly'), as $M= ab\dots c\rightarrow \pm nannbn\dots ncn=\pm nab\dots cn=\pm nMn$. Even more importantly, from the  Cartan-Dieudonn\'e theorem, rotations are the product of successive reflections. For instance, compounding the reflections in the hyperplanes defined by the unit vectors $n$ and $m$ results in a rotation in the plane defined by $n\wedge m$.
\begin{prop}[Rotations]\label{HGA_rot}
In Geometric Algebra, a vector `$a$' transforms under a rotation in the plane defined by $n\wedge m$ via successive reflection in hyperplanes determined by the unit vectors `$n$' and `$m$' as 
	\begin{equation}\label{in2rot}
	  a''=mnanm=: Ra\tilde{R},
	\end{equation}
where we have defined $R=mn$ and the tilde denotes the reversal of the order of the constituent vectors $\tilde{R}=nm$.
\end{prop}

\begin{thm}[Rotors and spinors]\label{HGA_thm_Rotor}
The object $R=mn$ generating the rotation in Eq. (\ref{in2rot}) is called a rotor. It satisfies $\tilde{R}R=R\tilde{R}=1$. Rotors themselves transform single-sidedly under further rotations, and thus form a multiplicative group under the geometric product, called the rotor group. Since $R$ and $-R$ encode the same rotation, the rotor group is a double-cover of the special orthogonal group, and is thus essentially the Spin group. Objects in Geometric Algebra that transform single-sidedly are called spinors, so that rotors are  normalised spinors. 
\end{thm}

\begin{cor}[Discrete spinor groups]\label{HGA_disspin}
	Discrete spinor groups are of even order.
\end{cor}

Higher multivectors transform in the above covariant, double-sided way as $ MN\rightarrow (RM\tilde{R})(R N \tilde{R})=RM\tilde{R}R N \tilde{R}=R(MN)\tilde{R}$.

	The Geometric Algebra of three dimensions $\Cl(3)$ spanned by three orthogonal unit vectors $e_1$, $e_2$ and $e_3$ contains three bivectors $e_1e_2$, $e_2e_3$ and $e_3e_1$ that square to $-1$, as well as the  highest grade object $e_1e_2e_3$   (trivector and pseudoscalar), which also squares to $-1$.
	\begin{equation}\label{in2PA}
	  \underbrace{\{1\}}_{\text{1 scalar}} \,\,\ \,\,\,\underbrace{\{e_1, e_2, e_3\}}_{\text{3 vectors}} \,\,\, \,\,\, \underbrace{\{e_1e_2=Ie_3, e_2e_3=Ie_1, e_3e_1=Ie_2\}}_{\text{3 bivectors}} \,\,\, \,\,\, \underbrace{\{I\equiv e_1e_2e_3\}}_{\text{1 trivector}}.
	\end{equation}

\begin{thm}[Quaternions and spinors of $\Cl(3)$]\label{HGA_quatBV}
The unit spinors $\lbrace 1,-Ie_1, -Ie_2, -Ie_3\rbrace$ of $\Cl(3)$ are isomorphic to the quaternion algebra $\mathbb{H}$. 
\end{thm}

This completes the background that we shall need for our proof of the Induction Theorem.

\section{Induction Theorem}\label{HGA_4D}

In this section, we show that every root system of rank 3 induces a root system in dimension 4. 

\begin{prop}[$O(4)$-structure of spinors and quaternions]\label{HGA_O4}
The space of $\Cl(3)$-spinors and quaternions have a 4D Euclidean signature.
\end{prop}
\begin{proof}
	For quaternions, this is given via  conjugation  defined by 
	$	\bar{q}=q_0-q_ie_i$,
	as
	$(p,q)=\frac{1}{2}(\bar{p}q+p\bar{q}),\,\,	|q|^2=\bar{q}q=q_0^2+q_1^2+q_2^2+q_3^2$.
	For a spinor $R=a_0+a_1Ie_1+a_2Ie_2+a_3Ie_3$, the norm is given by $R\tilde{R}=a_0^2+a_1^2+a_2^2+a_3^2$ (or via Theorem \ref{HGA_quatBV}), and the inner product is $(R_1, R_2)=\frac{1}{2}(R_1\tilde{R}_2+R_2\tilde{R}_1)$.
\end{proof}

\begin{lem}[Discrete quaternion groups give root systems]\label{HGA_quatgroups}
Any finite subgroup $G$ of even order in $\mathbb{H}$ is a root system.
\end{lem}

\begin{proof}
This is stated and proven in	\cite{Humphreys1990Coxeter}.
\end{proof}

\begin{lem}[Rank-3 Coxeter groups and finite even quaternion groups]\label{HGA_spinoreven}
	The spinors defined from any rank-3 Coxeter group are isomorphic to an even subgroup of the quaternions.
\end{lem}

\begin{proof}
From	Corollary	\ref{HGA_disspin}, the spinor group generated by a Coxeter group is discrete and even. Because of Theorem \ref{HGA_quatBV}, for a Coxeter group of rank 3 this spinor group is isomorphic to a finite even order quaternion group.
\end{proof}

\begin{lem}[Discrete spinor groups in 3D give 4D root systems]\label{HGA_spin4D}
	A discrete group of spinors in three dimensions is a four-dimensional root system.
\end{lem}

\begin{proof}
  Due to Lemma \ref{HGA_spinoreven}, the discrete spinor group is even and isomorphic to an even quaternion group. From Lemma	\ref{HGA_quatgroups} it follows that this is a root system. 
\end{proof}

\begin{thm}[Induction Theorem]
	Any rank-3 root system induces a  root system of rank 4.
\end{thm}

\begin{proof}
A root system in three dimensions gives rise to corresponding Coxeter reflections (Section \ref{HGA_Cox}), acting in Geometric Algebra as given by Eq. (\ref{in2refl}).
An even number of successive reflections yields spinors via Theorem \ref{HGA_thm_Rotor}, and from Corollary \ref{HGA_disspin}, this group is even. Via Lemma \ref{HGA_spin4D}, this even spinor group yields a root system in four dimensions. 

Alternatively, one can check the root system axioms from Definition \ref{DefRootSys} directly. The first is satisfied, since trivially for an element in a discrete spinor group its negative is also in the discrete spinor group by Corollary \ref{HGA_disspin}. One can easily check the axiom of closure of the root system under reflections by using Eq. (\ref{reflect}) in 4D using the Euclidean inner product in Proposition \ref{HGA_O4}. For  spinors $R_1$ and $R_2$ this amounts to 
$R_2\rightarrow R_2'=R_2-2(R_1, R_2)/(R_1, {R}_1) R_1 =R_2-((R_1\tilde{R}_2+R_2\tilde{R}_1) R_1/(R_1\tilde{R}_1)=  -R_1\tilde{R}_2R_1/(R_1\tilde{R}_1)$.
Closure of the root system is thus ensured by closure of the spinor group. This also has very interesting consequences for the automorphism group of these root systems, which contains two factors of the spinor group acting from the left and the right \cite{Dechant2013Platonic} (in this sense, the above closure under reflections amounts to a certain twisted conjugation). 
\end{proof}

\begin{thm}[Automorphism groups of the induced root systems]
	The automorphism groups of these induced root systems contain two factors of the respective spinor group (i.e. the root system itself), acting on the left and on the right (trivially, in the spinorial picture, as group multiplication in the spinor group, by group closure).
\end{thm}

\begin{thm}[Non-existence of a reduction theorem]
	Not every rank-4 root system can be induced by a rank-3 root system. 
\end{thm}

\begin{proof}
A counterexample is provided by $I_2(4)\times A_1 \times A_1$.
\end{proof}

In the above proof no reference was made to any particular root system, and it is thus valid in generality. However, the number of 
such root systems is limited, so that one can list the induced root systems on a case-by-case basis. Table \ref{tab_Symm} contains the list of 4D root systems that are induced by 3D root systems, as well as the results of later sections. It is particularly noteworthy that the exceptional four-dimensional root systems arise via our Clifford spinor construction; in particular, the fact that  irreducible root systems arise in this way. One might speculate philosophically whether their existence in some sense hinges on the accidentalness of this 3D Clifford spinor construction. 

\begin{ex}
	The simple roots of $A_1\times A_1 \times A_1$ can be taken as $\alpha_1=e_1$, $\alpha_2=e_2$ and $\alpha_3=e_3$. 
	Closure of these under reflections via Eq. (\ref{in2refl}) gives $(\pm 1, 0, 0)$ and permutations thereof, which are the 6 vertices of the root system, the octahedron.
Combining two reflections yields a spinor, so forming rotors according to $R_{ij}=\alpha_i\alpha_j$ gives, e.g. $R_{11}=\alpha_1^2=1\equiv(1,0,0,0)$, or $R_{23}=\alpha_2\alpha_3=e_2e_3=Ie_1\equiv (0,1,0,0)$, where we denote components as given in Theorem \ref{HGA_quatBV}. Explicit calculation of all cases generates the 8 permutations of $(\pm 1,0,0,0)$. When interpreted as a 4D polytope, these are the vertices of the 16-cell, which is the root system of 	$A_1\times A_1 \times A_1\times A_1$.
\end{ex}

\begin{ex}
Other specific examples are the inductions $A_3\rightarrow D_4$ (cuboctahedron to 24-cell), $B_3\rightarrow F_4$ and $H_3\rightarrow H_4$ (icosidodecahedron to 600-cell) \cite{Dechant2012CoxGA, Dechant2012AGACSE, Dechant2013Platonic}.
This demonstrates how exceptional phenomena in dimension four could arise from regular root systems (i.e. $A_n$ and $B_n$ for $n=3$), and hints that a spinorial view of their geometry can shed light on many of their applications.
\end{ex}	

\section{Spinors in  dimension two}\label{HGA_2D}

		The space of spinors $\psi=a+be_1e_2\equiv a+bI$ in two-dimensional Euclidean space is also two-dimensional, and has a natural Euclidean structure given by $\psi\tilde{\psi}=a^2+b^2$. This induces a rank-2 root system from any rank-2 root system in a similar way to the construction above. However, this construction does not yield any new root systems by the following theorem.

\begin{thm}[Self-duality of $I_2(n)$]
 Two-dimensional root systems are self-dual under the Clifford spinor construction.
\end{thm}

\begin{proof}
		A 2D root vector $\alpha_i=a_1e_1+a_2e_2$ is in bijection with a spinor by $\alpha_i\rightarrow \alpha_1\alpha_i=e_1\alpha_i =a_1+a_2e_1e_2=a_1+a_2I$ (taking $\alpha_1=e_1$ without loss of generality). This is the same as forming a spinor between those two root vectors.  The infinite family of two-dimensional root systems $I_2(n)$ is therefore self-dual.  The order of the Coxeter group $|W|$ matches the number of roots $|\Phi|$, even for the well-known crystallographic cases in 2D:  for instance,  for the $A_n$ family one has the general formulae $|W|=(n+1)!$ and $|\Phi|=n(n+1)$, with equality for  $n=2$, as $A_2=I_2(3)$. For $B_2=I_2(4)$, $|W|=2^n n!$ and $|\Phi|=2n^2$, such that equality holds for $n=2$.  For $G_2=I_2(6)$, one also has $|W|=|\Phi|=12$.
\end{proof}

\section{An octonion construction: eight-dimensional root systems}\label{HGA_Oct}

The crucial fact that our 3D and 2D ($nD$) constructions depended on was that Clifford algebra allowed us to construct a group of spinors with $2^{n-1}$ components, that -- when thought of as a set of vectors in $2^{n-1}$-dimensional space -- fulfilled the first axiom of a root system, i.e. that $\alpha$ and $-\alpha$ are contained in the set, by construction. This $2^{n-1}$-dimensional space also had a Euclidean metric. The second part of the proof that these sets are in fact root systems was then to show closure under reflections in these roots as defined by the Euclidean metric. For instance, this is no longer the case in 4D, where the space of spinors is 8-dimensional.  However, here we follow an analogous construction in terms of octonions, that does yield root systems in 8D. 

The octonions $\mathbb{O}$ are a generalisation of the quaternions, complex and real numbers. It is a (non-associative division) algebra with one scalar and seven imaginary units that satisfy certain commutativity properties. In analogy with complex and quaternionic conjugation, there is an octonionic conjugation, which defines an 8-dimensional Euclidean metric. 

In order to mimic our Clifford construction above, we take the simple roots of the 3D root systems to be along the first three imaginary units $i$, $j$ and $k$, i.e. the simple roots of $A_1^3$ would be $\alpha_1=i$, $\alpha_2=j$ and $\alpha_3=k$. We then consider the closure of this set of octonions under octonion multiplication (here 16 octonions, namely the 8 unit octonions and their negatives). Again this gives a set (group) of octonions which contains $\pm \alpha_i$ by construction. It is also straightforward to show that these sets of octonions are closed under reflections with respect to the 8-dimensional Euclidean metric, as this reflection essentially amounts to octonion multiplication and  the above procedure of taking the closure under octonion multiplication therefore guarantees closure under reflections. The resulting set of vectors is therefore a root system in eight dimensions (here $A_1^8$). However, the resulting root system in each case is just the direct sum of two copies of the corresponding induced 4D root system, e.g. $H_3$ induces $H_4$ in 4D and $H_4 \oplus H_4$ in 8D. The argument in the paper is general and does not make reference to any particular root system. Table \ref{tab_Symm} therefore summarises the results for each of the limited number of root systems for illustrative purposes. 

\begin{table}
	\caption{\label{tab_Symm} Summary of the induced root systems in two, four and eight dimensions.}
\begin{centering}\begin{tabular}{l l}
\br
start root system&induced root system
\tabularnewline
\mr
$I_2(n)$&$I_2(n)$
\tabularnewline
\mr
$A_1\oplus I_2(n)$&$I_2(n)\oplus I_2(n)$
\tabularnewline
$A_3$&$D_4$
\tabularnewline
$B_3$&$F_4$
\tabularnewline
$H_3$&$H_4$
\tabularnewline
\mr
$A_1\oplus I_2(n)$&$I_2(n)\oplus I_2(n)\oplus I_2(n)\oplus I_2(n)$
\tabularnewline
$A_3$&$D_4\oplus D_4$
\tabularnewline
$B_3$&$F_4\oplus F_4$
\tabularnewline
$H_3$&$H_4\oplus H_4$
\tabularnewline
\br
\end{tabular}\par\end{centering}
\end{table}

\section{Conclusions}\label{HGA_Concl}
We have shown how via a Clifford spinor construction, any root system of rank 3 induces a root system in four dimensions. This was done in full generality without reference to any particular root system. However, since the number of root systems in 3D is finite, one gets a concrete list of cases. One finds that the 4D root systems induced in this way contain mostly the exceptional root systems that generate the exceptional Coxeter groups $D_4$ (triality), $F_4$ (largest crystallographic group in 4D) and $H_4$ (largest non-crystallographic group). This construction  therefore offers a novel perspective on exceptional phenomena in four dimensions and the peculiar structure of their automorphism groups, in terms of spinorial geometry. 

Via an analogous construction in terms of octonions, each such 3D root system also induces a root system in eight dimensions; however, this root system is reducible and in fact consists of two copies of the induced 4D root system in each case.  In the two-dimensional case, root systems (i.e. $I_2(n)$) were shown to be self-dual. This spinorial view sheds light on the peculiarities of root systems, in particular certain rank-4 root systems and their automorphism groups (see \cite{Dechant2012CoxGA, Dechant2013Platonic}), and more generally, opens up a new field of study in the spinorial geometry of root systems. This could be particularly interesting for applications in high energy physics, where the rank-4 groups are ubiquitous.


\ack
I would like to thank my family and friends for their support and David Hestenes, Eckhard Hitzer, Anthony Lasenby, Joan Lasenby, Reidun Twarock, Mike Hobson and C\'eline B\oe hm for helpful discussions.


\end{document}